\documentclass{amsart}

\usepackage[utf8]{inputenc}
\usepackage[T1]{fontenc}
\usepackage{amsmath}
\usepackage{amsthm}
\usepackage{amssymb}
\usepackage{amsfonts}
\usepackage{booktabs}
\usepackage{algorithm}
\usepackage{algorithmic}
\usepackage[switch]{lineno}
\usepackage{placeins}
\usepackage{enumitem}
\usepackage{tabularx}

\usepackage[footskip=1cm,margin=1in]{geometry}

\usepackage{soul}
\usepackage{url}
\usepackage[hidelinks]{hyperref}
\usepackage[utf8]{inputenc}
\usepackage[small]{caption}
\usepackage{graphicx}
\usepackage{amsmath}
\usepackage{amsthm}
\usepackage{booktabs}
\usepackage{algorithm}
\usepackage{algorithmic}

\usepackage{natbib}
\setcitestyle{numbers,open={[},close={]}}

\makeatletter
\renewcommand{\section}{%
  \@startsection{section}{1}%
    \z@%
    {-1.5\baselineskip}%
    {0.8\baselineskip}%
    {\normalfont\scshape\centering}%
}

\renewcommand{\subsection}{%
  \@startsection{subsection}{2}%
    \z@%
    {-1.2\baselineskip}
    {0.6\baselineskip}
    {\normalfont\bfseries}%
}
\makeatother

\usepackage{cleveref}
\crefname{figure}{figure}{figures}
\hypersetup{
		pdfencoding=auto, 
		psdextra,
		colorlinks=true,
		citecolor=green!40!black,
		linkcolor=red!50!black,
		urlcolor=blue!80!black
	}
\usepackage{subcaption}

\setlist[description]{labelwidth=1cm}

\let\oldemph\emph
\renewcommand{\emph}[1]{\textcolor{blue!60!black}{\oldemph{#1}}}

\usepackage{fancyhdr}
\usepackage{mathtools}
\usepackage{subcaption}
\usepackage{booktabs}
\usepackage{multirow}
\usepackage[export]{adjustbox}
\usepackage{xargs}

\usepackage{booktabs}
\usepackage{nicefrac}
\usepackage{chngpage}
\usepackage{todonotes}
\usepackage{cleveref}

\usepackage{tikz}
\usetikzlibrary{cd}

\newtheorem{proposition}{Proposition}

\theoremstyle{definition}
\newtheorem{example}{Example}
\newtheorem{remark}{Remark}
\newtheorem{definition}{Definition}
\newtheorem{lemma}{Lemma}

\newcommand{\normphi}{{{\mathrm{norm}\hbox{-}\phi}}}

\usepackage{todonotes}

\newcommand{\Dkl}[2]{\mathrm{D}_{\mathrm{KL}}\!\left(#1\,\Vert\,#2\right)}

\makeatletter 
\def\WFfill{\par 
    \ifx\parshape\WF@fudgeparshape 
    \nobreak 
    \ifnum\c@WF@wrappedlines>\@ne 
    \advance\c@WF@wrappedlines\m@ne 
    \vskip\c@WF@wrappedlines\baselineskip 
    \global\c@WF@wrappedlines\z@ 
    \fi 
    \allowbreak 
    \WF@finale 
    \fi 
} 
\makeatother


\title{Allocation Proportionality of OWA--Based Committee Scoring Rules}

\author{Daria Boratyn}
\author{Dariusz Stolicki}

\address{Jagiellonian Center for Quantitative Political Science, Jagiellonian University, ul. Reymonta 4, Kraków, Poland}
\email{daria.boratyn@uj.edu.pl}
\email{dariusz.stolicki@uj.edu.pl}

\thanks{This research has been funded under the Polish National Center for Science grant no.~2019/35/B/HS5/03949 and the Jagiellonian University Excellence Inititative, QuantPol Center project.}

\begin{document}

\maketitle

\begin{abstract}
While proportionality is frequently named as a desirable property of voting rules, its interpretation in multiwinner voting differs significantly from that in apportionment. We aim to bridge these two distinct notions of proportionality by introducing the concept of allocation proportionality, founded upon the framework of party elections, where each candidate in a multiwinner election is assigned to a party. A~voting rule is allocation proportional if each party's share of elected candidates equals that party's aggregate score. Recognizing that no committee scoring rule can universally satisfy allocation proportionality in practice, we introduce a new measure of allocation proportionality degree and discuss how it relates to other quantitative measures of proportionality. This measure allows us to compare OWA-based committee scoring rules according to how much they diverge from the ideal of allocation proportionality. We present experimental results for several common rules: SNTV, $k$-Borda, Chamberlin-Courant, Harmonic Borda, Proportional $k$-Approval Voting, and Bloc Voting.
\end{abstract}

\section{Preliminaries}

In modern social choice, proportionality is widely (though not universally) regarded as an important component of fairness in preference aggregation \cite{Baker96}. It is a principle positing that each cohesive block of voters should be represented in the election outcome in proportion to its size. Obviously, this general statement admits multiple interpretations (what does it mean that a voting block is cohesive? what does it mean that it is represented? etc.), which translate into variety of formal definitions. In particular, the interpretation of proportionality in multiwinner voting differs significantly from that in apportionment (see Sec.~\ref{sec:prop}). To reconcile these two distinct conceptions of proportionality, in this paper we introduce a new notion of \emph{allocation proportionality}, based on the framework of party elections.

\subsection{Motivation}

There is a substantial gap between voting methods studied in computational social choice and those actually used in real-life political elections. Apart from STV, SNTV, and few marginal instances of the use of $k$--Borda, no ranked-choice (i.e., ordinal) multiwinner voting system is used in national political elections -- and even STV and SNTV are regarded as outliers among real-life electoral systems. For this reason, while the social choice properties of multiwinner voting rules have been extensively studied, we know almost nothing about their political effects. This paper seeks to bridge that gap by positioning OWA--based rules -- an important class of ordinal multiwinner voting rules -- on the fundamental axis describing electoral systems: from proportional to majoritarian.

The other motivation for this paper arises from one of the fundamental difficulties in defining proportionality for ordinal elections: we have no \textit{a priori} notion of which voter groups should be represented proportionally -- and even what does it mean for a group to be represented if candidates are ranked rather than approved. In most real-life political elections, we have predefined political parties with well-defined electorates, but this is untrue for most elections using ordinal voting. Thus, proportionality axioms are usually formulated in terms of cohesive voting blocks. Unfortunately, most definitions of cohesive voting blocks are rather restrictive, given that such blocks rarely appear in real-life data \cite{BrillEtAl23,McCuneGraham-Squire24}. Proportionality axioms based on such definitions tell us little about the behavior of voting systems in typical conditions. This paper seeks to solve the problem by proposing a widely applicable definition and measure of proportionality that uses experimental data with predefined parties.

\subsection{Contribution}

\textbf{First.} We introduce a new concept of allocation proportionality for OWA--based voting rules, bridging the gap between definitions based on representation of cohesive voting blocks and those focusing on apportionment proportionality.

\noindent \textbf{Second.} We introduce quantitative measures of allocation proportionality degree, corresponding to those used in apportionment proportionality.

\noindent \textbf{Third.} We compare common OWA--based voting rules, viz., SNTV, $k$--Borda, Chamberlin--Courant, Harmonic Borda, $k$--PAV, and Bloc Voting, measuring their allocation proportionality degree.

\noindent \textbf{Fourth.} We find that SNTV and $k$--PAV attain the best approximation of allocation proportionality, while $k$--Borda and Bloc Voting are consistently biased in favor of large parties. Chamberlin--Courant usually performs quite well, but -- especially if the number of parties is small while the committee size is large -- it can be highly biased in favor of small parties.

\subsection{Notation}

We end this section by introducing some notation that will be used throughout the paper:

\begin{description}
    \item[$\mathcal{P}$] For any set $X$, we denote its power set by $\mathcal{P}(X)$.
    \item[${[n]}$] For any ordinal $n$, let $[n] := \{1, \dots, n\}$ be the set of ordinals from $1$ to $n$.
    \item[$\Delta_{d}$] For any $d \in \mathbb{N}_{+}$, let $\Delta_d := \{ \mathbf{x} \in \mathbb{R}_{+}^{d} : \sum_{i \in [d]} x_i = 1\}$ be the $d$-vertex unit simplex.
    \item[$\lVert\cdot\rVert_{p}$] For any vector $\mathbf{x} \in \mathbb{R}^{d}$, where $d \in \mathbb{N}_{+}$, we denote its $L_p$ norm by $\lVert\mathbf{x}\rVert_{p}$, where $p \in [1, \infty) \cup \{\infty\}$.
    \item[$L$] For any set $X$, we denote the set of linear orders thereon by $L_X$.
\end{description}

We use $n$ for the number of voters, $m$ for the number of candidates, and $k$ for the committee size. In general, we use capital letters to denote sets and random variables, and boldface font for vectors, votes, and profiles.

\section{Proportionality Axioms and Quantitative Measures} \label{sec:prop}

The concept of proportional representation has been developed primarily with single-choice voting in mind \cite[see, e.g.,][]{BalinskiYoung01,Pukelsheim17}. In that context, it is quite clear what the standard of proportionality is: it consists of identity between vote shares and seat shares (we will refer to it as \emph{apportionment proportionality}). The only problem lies in that standard not being attainable for any system with discrete atomic seats. Hence, the scholarly debate about apportionment proportionality has focused primarily on how to quantify the inevitable deviations from the proportionality standard \cite{TaageperaGrofman03}, and which particular deviation objective a proportional representation system should minimize \cite{Wada16}.

In ordinal elections -- where a ballot consists of a linear ordering of candidates -- formalizing `proportionality' is more challenging \cite{BrillPeters23}. Under ordinal elections, there is no natural measure of the `number of votes' with which to combine seat shares. Moreover, unlike real-life single-choice systems, ordinal elections tend to lack explicit party groupings, so axioms have been developed to identify when a group of like-minded voters ``deserves'' representation.

\subsection{Proportionality Axioms}

The first definition of proportionality for ordinal voting systems has been proposed by \citet{Dummett84}, who formulated the axiom of \emph{Proportionality for Solid Coalitions (PSC)}. It requires any sufficiently large group of voters with ballot prefixes that are identical up to permutation, i.e., who rank the same set of candidates above all others, to receive proportional representation. Formally, for a quota $q$ (typically Hare $n/k$ or Droop $\lfloor n/(k+1) \rfloor + 1$), a group at least $\ell q$ in size that unanimously ranks a set of candidates of cardinality $\ell$ above all others must have at least $\ell$ members of the committee elected.

Unlike apportionment proportionality, PSC is actually attainable in practice. Indeed, one of the earliest known multiwinner ordinal voting systems, \emph{single-transferable vote (STV)} \cite{Hare59}, satisfies PSC \cite{Tideman95,LacknerSkowron20}. Nevertheless, PSC's demand for unanimity in ranking limits its robustness \cite{BrillPeters23}. Even slight deviations in the top-ranked candidates among coalition members can void its representation guarantee.

In approval voting, \citet{AzizEtAl17} introduced \emph{Justified Representation (JR)}, which requires any group of at least $n/k$ voters approving a common candidate to have at least one of their approved candidates elected. \emph{Extended Justified Representation (EJR)} strengthens this by requiring that groups large enough to deserve $\ell$ seats who approve at least $\ell$ common candidates must see at least one member approve $\ell$ elected candidates. \emph{Proportional Justified Representation (PJR)} \cite{Sanchez-FernandezEtAl17} is an intermediate axiom requiring that such a group sees $\ell$ approved candidates elected, not necessarily per individual. These axioms form a hierarchy: EJR $\Rightarrow$ PJR $\Rightarrow$ JR.


\citet{BrillPeters23} introduced robust variants of PSC and EJR that relax the unanimity requirements. These new axioms allow near-cohesive groups to gain representation and are polynomial-time verifiable. Their robust ordinal PSC variant is strictly stronger than classical PSC and is not satisfied by STV. For clustering generalizations of their model, see \cite{AzizEtAl24,KellerhalsPeters24}.

\citet{AzizLee21} generalized PSC to weak preference orders and proposed the Expanding Approvals Rule (EAR), which satisfies generalized PSC while addressing STV’s non-monotonicity. STV generalizations satisfying generalized PSC have been considered in \cite{DelemazurePeters24}.

\subsection{Quantitative Measures}

Empirical work (e.g., \cite{BrillEtAl23,McCuneGraham-Squire24}) shows solid coalitions rarely occur in real-world data, motivating the shift toward quantitative or robust proportionality assessments.

The first non-apportionment measure of proportionality in computational social choice has been the \emph{proportionality degree}, introduced by \citet{Skowron18} for approval voting systems. Proportionality degree was defined using the concept of a \emph{$k$-proportionality guarantee}: a function of voter group size and committee size such that if the intersection of that group's approval sets exceeds the guarantee, the group's average number of approved representatives will always do so as well. A proportionality degree of a rule is the best possible proportionality guarantee that holds irrespectively of the size of the committee \cite{Skowron18}.

A similar approach to quantifying proportionality has been taken by \citet{Janson18}, who analyzed threshold guarantees for PSC, JR, EJR, and PJR, comparing rules by coalition sizes required to guarantee representation.

\citet{BardalEtAl25} quantify ordinal proportionality by first defining a weakened, parametrized version of PSC (and several other axioms), and then finding the smallest parameter value for which the parameterized axiom is satisfied in a given election. They provide experimental results, focusing primarily on STV, and note certain similarities between their PSC value and Jefferson--D'Hondt apportionment method.

\section{Theoretical Framework: Party Elections}

The main difficulty underlying most attempts at defining or measuring proportionality in the context of ordinal or approval voting lies in the fact that cohesive voting groups are not predefined. However, if our focus is on proportionality as a property of voting rules, rather than as a quantity to be measured for particular, real-life elections, we can avoid that problem by simply predefining the groups and analyzing elections in which some form of cohesive behavior is enforced by the statistical culture. For that purpose, we use the framework of \emph{party elections}, first introduced by \citet{BoratynEtAl25}.

\begin{definition}[\sc Party Election] \label{def:partyElec}
    A \emph{party election} ($C$, $P$, $a$, $\mathcal{D}_{C}$, $\mathbf{w}$) is a tuple consisting of a set of \emph{candidates}, $C$; a set of \emph{parties}, $P$; a surjective \emph{affiliation function}, $a: C \rightarrow P$, that maps each candidate to a party; a \emph{set of admissible votes}, $\mathcal{D}_{C}$; and a \emph{profile} (i.e., a collection of votes), $\mathbf{w}$.
\end{definition}

Intuitively, a party election is an election endowed with an additional structure: every candidate is mapped to some party. A party election is ordinal if $D_C = L_C$, i.e., the set of admissible votes equals the set of linear orders on $C$. This will be assumed throughout the remainder of the paper.


If party elections are to be regarded as approximations of real-life political elections, one additional complication is in order. In most countries, a `legislative election' is actually not a single election in the sense of Definition \ref{def:partyElec}, but an ensemble of parallel elections held in distinct (though sometimes overlapping) \emph{electoral districts}. This is important insofar as proportionality is concerned, since -- depending on the distribution of preferences across districts -- use of multiple districts can either improve proportionality or aggravate disproportionality. Formally, incorporating districts into our model is quite simple:

\begin{definition}[\sc Multi-District Party Election] \label{def:mdPartyElec}
    A \emph{multi-district party election} is a sequence of party elections, $\mathcal{E}_1, \dots, \mathcal{E}_{c}$, where $c$ is the number of districts and $\mathcal{E}_{i} := (C_i, P, a_i, \mathcal{D}_{C_i}, \mathbf{w}_{i})$ for every $i \in [c]$, such that all districts share the same set of parties.
\end{definition}

\section{Committee Scoring Rules}

A \emph{multiwinner voting rule} is any function that maps a candidate set $C$, a profile $\mathbf{w}$, and a committee size $k \in [m]$ to a set of $k$-subsets of $C$ that tie as winning committees. However, our focus in this chapter is only on a subset of all multiwinner voting rules, namely on the class of \emph{OWA--based committee scoring rules}. They are a subclass of \emph{committee scoring rules}, introduced by \citet{FaliszewskiEtAl16,ElkindEtAl17a} and treated in detail by \citet{FaliszewskiEtAl19a}.

For a linear order $\prec$ and a candidate $i$, let $\mathrm{pos}_{\prec}(i) := |\{j \in C: i \preceq j\}|$ denote the position of $i$ in $\prec$, i.e., the number of candidates $j \in C$ such that $i \preceq j$.

\begin{definition}[\sc Position Vector]
    For a $k$-committee $S \subseteq C$ and a linear order $\prec$, the \emph{position vector}, $\mathbf{p}_{\prec}(S)$, is a vector of positions of candidates of $S$ in $\prec$ permuted so that its coordinates are sorted in the increasing order:
    \[
        \mathbf{p}_{\prec}(S) = \bigl(p_{\prec, 1},\dots, p_{\prec, k}\bigr)
    \]
        where
    \[
        \{p_{\prec, 1}, \dots, p_{\prec, k}\} := \{\mathrm{pos}_{\prec}(i): i \in S\}, \; p_{\prec, 1} < \cdots < p_{\prec, k}.
    \]
\end{definition}

We denote the set of all position vectors for a committee of $k$ out of $m$ candidates (i.e., the set of increasing functions $[k] \rightarrow [m]$) by $[m]_k$.

\begin{definition}[\sc Partial Order on Position Vectors] \label{def:partialOrder}
    Let $\trianglelefteq$ be a partial order on $[m]_k$ such that for any position vectors $\mathbf{p}, \mathbf{q} \in [m]_k$ we have $\mathbf{p} \trianglelefteq \mathbf{q}$ (respectively $\mathbf{p} \triangleleft \mathbf{q}$) if and only if for every $i \in [k]$ we have $p_i \ge q_i$ (respectively $p_i > q_i$).
\end{definition}

\begin{definition}[\sc Committee Scoring Function \cite{ElkindEtAl17a}]
    A \emph{committee scoring function} for a committee of $k$ out of $m$ candidates is a function $f_{m,k}: [m]_k \rightarrow \mathbb{R}_{\ge 0}$ such that $\mathbf{p} \trianglelefteq \mathbf{q}$ implies $f_{m,k}(\mathbf{p}) \le f_{m,k}(\mathbf{q})$ for every $\mathbf{p}, \mathbf{q} \in [m]_k$.
\end{definition}

\begin{definition}[\sc Committee Scoring Rules \cite{ElkindEtAl17a}]
    Let $f := (f_{m,k})_{k \le m}$ be a family of committee scoring functions. A \emph{committee scoring rule} $\mathcal{R}_{f}$ is a multiwinner social choice function that maps a candidate set $C$, a profile $\mathbf{w}$, and a committee size $k \in [m]$ to the set of all $k$-subsets of $C$ maximizing $\sum_{\prec \in \mathbf{w}} f_{m,k} \left(\mathbf{p}_{\prec}(S)\right)$
\end{definition}

The class of \emph{OWA--Based Committee Scoring Rules} was first proposed by \citet{SkowronEtAl16a}. The intuition behind OWA-based voting rules is to define committee scoring functions in terms of weighted sums of individual candidate scores while introducing flexibility in how voter satisfaction is defined. Instead of assuming that each voter evaluates a committee solely based on, for example, their most preferred or least preferred member within it, OWA-based rules use an \emph{ordered weighted averaging} (OWA) operator to interpolate between these extremes. For each vote $\prec$, the positional scoring vector is used to assign a score to every candidate as a function of that candidate’s position in $\prec$. A committee evaluation score is then obtained by aggregating committee member scores using the OWA operator defined by OWA weight vector $\mathbf{z}$, which weights the voter’s $l$-th most preferred committee member by $z_l$. The winning committee is the $k$-subset whose total OWA-aggregated satisfaction, summed across all voters, is maximal. This framework captures a wide range of aggregation attitudes: from highly egalitarian (e.g., considering only the worst-ranked member) to highly utilitarian (e.g., summing scores of all members), depending on the choice of the OWA weight vector.

\pagebreak

\begin{definition}[\sc OWA--Based Scoring Rule]
An \emph{OWA--based scoring rule} induced by an \emph{OWA weight vector} $\mathbf{z} \in [0, 1]^{k}$ and a \emph{scoring vector} $\mathbf{s}\in [0, 1]^{m}$, $\mathcal{R}^{\mathrm{OWA}}_{\mathbf{s}, \mathbf{z}}$, is a social choice function that maps a party election to $[C]^{k}$, the set of all $k$--subsets of $C$, maximizing over committee $S \in [C]^{k}$ the expression $\sum_{\prec \in \mathbf{w}} \sum_{i=1}^k z_i s_{p_{\preceq,i}(S)}$, where $p_{\preceq,i}(S)$ is the $i$--th coordinate of position vector $\mathbf{p}_{\preceq}(S)$ \citep{SkowronEtAl16a, FaliszewskiEtAl19a}.
\end{definition}

To define common OWA--based rules, let us introduce some standard scoring and OWA weight vectors:

\begin{description}[labelwidth=3em, leftmargin=\dimexpr\labelwidth+\labelsep\relax]
\item[{$\mathbf{u}_{k,l}$}]
Let the $l$-dimensional \emph{$k$--approval vector} $\mathbf{u}_{k,l}$, $l \ge k$, be a vector of $k$ ones followed by $l-k$ zeroes:
$$\mathbf{u}_{k,l} := (\overbrace{1, \dots, 1}^{k}, \overbrace{0, \dots, 0}^{l-k}).$$
\item[{$\mathbf{b}_{k}$}]
Let the $k$-dimensional \emph{Borda vector} $\mathbf{b}_{k}$ be defined as $$\mathbf{b}_{k} := \frac{(k-1, k-2, \dots, 0)}{k-1}.$$
\item[{$\mathbf{h}_{k}$}]
Let the $k$-dimensional \emph{harmonic vector} $\mathbf{h}_{k}$ be defined as $$\mathbf{h}_{k} := \left(1, \frac{1}{2}, \dots, \frac{1}{k}\right).$$
\end{description}

The rules belonging to the class of OWA--based scoring rules include:

\begin{definition}[SNTV]
    The Single Non-Transferable Vote (SNTV) rule, $$\mathcal{R}_{k}^{\mathrm{SNTV}} := \mathcal{R}^{\mathrm{OWA}}_{\mathbf{u}_{1,m}, \mathbf{u}_{1,k}},$$ is an OWA--based rule induced by the $1$-approval scoring vector $\mathbf{u}_{1,m}$ and the $1$-approval OWA--vector $\mathbf{u}_{1,k}$, i.e., where the members of the winning committee maximize the number of votes where they were ranked first.
\end{definition}

\begin{definition}[\sc $k$--Borda] 
    The $k$--Borda rule, $$\mathcal{R}_{k}^{\mathrm{Borda}} := \mathcal{R}^{\mathrm{OWA}}_{\mathbf{b}_{m}, \mathbf{u}_{k,k}},$$ is an OWA--based rule induced by the Borda scoring vector $\mathbf{b}_{m}$ and the $k$-approval OWA--vector $\mathbf{u}_{k,k}$, i.e., where the members of the winning committee maximize the sum of their Borda scores \citep{deBorda81,FaliszewskiEtAl17}.
\end{definition}

\begin{definition}[\sc Bloc Voting]
    The Bloc Voting rule, $$\mathcal{R}_{k}^{\mathrm{BV}} := \mathcal{R}^{\mathrm{OWA}}_{\mathbf{u}_{k,m}, \mathbf{u}_{k,k}},$$ is an OWA--based rule induced by the $k$-approval scoring vector $\mathbf{u}_{k,m}$ and the $k$-approval OWA--vector $\mathbf{u}_{k,k}$, i.e., where the members of the winning committee maximize the number of votes where they were ranked within $k$ first candidates.
\end{definition}

\begin{definition}[\sc Chamberlin--Courant (CC)]
    The Chamberlin--Courant (CC) rule, $$\mathcal{R}_{k}^{\mathrm{CC}} := \mathcal{R}^{\mathrm{OWA}}_{\mathbf{b}_{m}, \mathbf{u}_{1,k}},$$ is an OWA--based rule induced by the Borda scoring vector $\mathbf{b}_{m}$ and the $1$--approval OWA--vector $\mathbf{u}_{1,k}$, i.e., where the winning committee maximizes the sum of Borda scores of the highest--ranking member \citep{ChamberlinCourant83}.
\end{definition}

\begin{definition}[\sc Harmonic-Borda]
    The Harmonic Borda rule, $$\mathcal{R}_{k}^{\mathrm{HB}} := \mathcal{R}^{\mathrm{OWA}}_{\mathbf{b}_{m}, \mathbf{h}_{k}},$$ is an OWA--based rule induced by the Borda scoring vector $\mathbf{b}_{m}$ and the harmonic OWA--vector $\mathbf{h}_{k}$, i.e., where the winning committee maximizes the sum of committee scores defined as the sum of the Borda scores of committee members with harmonically decreasing weights \citep{FaliszewskiEtAl17}.
\end{definition}

\begin{definition}[\sc Proportional $k$--Approval Voting]
    The Proportional $k$--Approval rule, $$\mathcal{R}_{k}^{\mathrm{PAV}} := \mathcal{R}^{\mathrm{OWA}}_{\mathbf{u}_{k,m}, \mathbf{h}_{k}},$$ is an OWA--based rule induced by the $k$-approval scoring vector, $\mathbf{u}_{k,m}$, and the harmonic OWA--vector $\mathbf{h}_{k}$, i.e., where the winning committee maximizes the sum of harmonic scores, $H_j$, where $j$ is the number of committee members included within the voter's $k$ highest--ranked candidates \citep{Kilgour10,Thiele95}. Note that this rule differs from classical PAV (Thiele) rule by requiring each voter to approve exactly $k$ candidates.
\end{definition}

\section{Allocation Proportionality}

Even with predefined parties, there remains the more fundamental of the two problems involved in applying apportionment proportionality in the ordinal voting context: the lack of a natural equivalent to the number of votes. For such an equivalent to exist, we would need a total preorder $\curlyeqprec$ on the set of profiles, $V_{C}$, such that the quotient poset of $(V_{C}, \curlyeqprec)$ is order-embeddable into $(\mathbb{R}, \le)$, and $\curlyeqprec$ extends the partial ordering on position vectors $\trianglelefteq$ given by Definition \ref{def:partialOrder} in such a manner that every voting rule monotonic with respect to $\trianglelefteq$ (i.e., in the sense of \citet{Fishburn77}) is likewise monotonic with respect to $\curlyeqprec$. In particular, all committee scoring rules would need to be monotonic with respect to $\curlyeqprec$.

However, consider the following example: fix any $m > 4$, any $k \ge 3$, and any candidate $i \in C$. For every $j \in [m]$ let $\prec_{j} \in L_{C}$ be such that $\mathrm{pos}_{\prec_j}(i) = j$. Finally, let $\mathbf{w}_{1}:= \{\prec_1, \prec_{k+1}\}$ and $\mathbf{w}_{2}:= \{\prec_2, \prec_{k-1}\}$. Under $k$-Borda, candidate $i$ is better off under $\mathbf{w}_{2}$ than $\mathbf{w}_{1}$, implying that $\mathbf{w}_{1} \curlyeqprec \mathbf{w}_{2}$. But under $k$-PAV, candidate $i$ is better off under $\mathbf{w}_{1}$ than under $\mathbf{w}_{2}$, implying that $\mathbf{w}_{2} \curlyeqprec \mathbf{w}_{1}$. As this is a contradiction, it follows that no such total preorder $\curlyeqprec$ exists.

Nevertheless, while no universal measure of the number of votes exists for ordinal voting, rule-specific measures can be defined. In particular, for any OWA--based rule $\mathcal{R}$ induced by a scoring vector $\mathbf{s}$, we can define the \emph{aggregate party score}:

\begin{definition}[\sc Aggregate Party Score] \label{def:aps}
    Let $\mathcal{R}$ be any OWA--based rule induced by a scoring vector $\mathbf{s}$, and let $(C, P, a, L_{C}, \mathbf{w})$ be a party election. For every party $i \in P$, its \emph{aggregate party score} under $\mathcal{R}$ is the sum of its candidate scores divided by the sum of all candidate scores, i.e.,
    \begin{equation} \label{eq:psi-def}
        \psi_i^{\mathcal{R}} := \frac{\sum_{\prec \in \mathbf{w}} \sum_{j \in a^{-1}(i)} s_{\mathrm{pos}_{\prec}(j)}}{\sum_{\prec \in \mathbf{w}} \sum_{j \in C} s_{\mathrm{pos}_{\prec}(j)}} = \frac{\sum_{\prec \in \mathbf{w}} \sum_{j \in a^{-1}(i)} s_{\mathrm{pos}_{\prec}(j)}}{|\mathbf{w}| \lVert\mathbf{s}\rVert_1}.
    \end{equation}
\end{definition}

For multi-district elections, aggregate party scores are obtained by averaging over districts.

With aggregate party scores as our counterpart to vote shares, we can finally introduce our definition of \emph{allocation proportionality}:

\begin{definition}[\sc Allocation Proportionality]
    An OWA--based voting rule is \emph{allocation proportional} for a given party election if for every $i \in P$ the share of committee seats obtained by candidates of the $i$-th party equals the aggregate party score of that party.
\end{definition}

Let us establish that our aggregate party score $\psi$ (Def.~\ref{def:aps}) is a good \emph{vote surrogate} for ordinal profiles by showing that it satisfies the following desiderata: 
(i) \emph{anonymity} and \emph{neutrality}; 
(ii) \emph{additivity} in the electorate size; 
and (iii) \emph{monotonicity}: if each candidate of party $i$ is (weakly) moved up in every ballot w.r.t.\ $s$, then $\psi_i$ does not decrease.

\begin{lemma}[\sc Basic properties of $\psi$]
\label{lem:psi-properties}
Fix an ordinal party election $(C,P,a,L_C,w)$ and a scoring vector $\mathbf{s}$. Fix any OWA--based rule $\mathcal{R}$ induced by $\mathbf{s}$. Then:
\begin{enumerate}
\item \textbf{Normalization:} $\sum_{i\in P} \psi_i^\mathcal{R}=1$.
\item \textbf{Anonymity and neutrality:} $\psi_i^\mathcal{R}$ is invariant under permuting voters and under any relabeling of parties and candidates that preserves $a$.
\item \textbf{Additivity (reinforcement):} For two profiles $\mathbf{w}, \mathbf{w}'$ on the same $(C,P,a)$,
\[
\psi_i^\mathcal{R}(C, P, a, L_C, \mathbf{w} \cup \mathbf{w}') \;=\; 
\frac{|\mathbf{w}|\,\psi_i^\mathcal{R}(C,P,a,L_C,\mathbf{w})+|\mathbf{w}'| \, \psi_i^\mathcal{R}(C,P,a,L_C,\mathbf{w}')}{|\mathbf{w}|+|\mathbf{w}'|}.
\]
In particular, if $\psi_i^\mathcal{R}(\mathbf{w})=\psi_i^\mathcal{R}(\mathbf{w}')$, then $\psi_i^\mathcal{R}(\mathbf{w} \cup \mathbf{w}')=\psi_i^\mathcal{R}(\mathbf{w})$.
\item \textbf{Monotonicity in positions:} If in going from $\mathbf{w}$ to $\mathbf{w}'$ every candidate of party $i$ weakly improves their position in every vote (w.r.t.\ $\mathbf{s}$), then $\psi_i^\mathcal{R}(\mathbf{w}')\ge \psi_i^\mathcal{R}(\mathbf{w})$.
\item \textbf{Independence from OWA weights:} $\psi_i^\mathcal{R}$ depends only on the scoring vector $\mathbf{s}$ and not on the OWA vector of the evaluated rule.
\end{enumerate}
\end{lemma}

\begin{proof}
(1) Immediate from Definition~\ref{def:aps} since the denominator equals the sum of numerators over $i\in P$.
(2) Both follow because $\psi$ is defined by summing position-based scores and then normalizing;
these operations commute with voter permutations and candidate/party relabelings that preserve $a$.
(3) For each party $i$, the numerator in~\eqref{eq:psi-def} adds over $\mathbf{w}$ and $\mathbf{w}'$,
and the denominator scales by $|\mathbf{w}|+|\mathbf{w}'|$; dividing yields the stated convex combination.
(4) If party $i$'s candidates weakly improve in each vote (w.r.t.\ $\mathbf{s}$), then the numerator of $\psi_i^\mathcal{R}$ (a sum of nonnegative $\mathbf{s}$-scores) weakly increases while the denominator is fixed, hence $\psi_i^\mathcal{R}$ weakly increases.
(5) OWA vector appears only in the \emph{evaluation of committees}, not in~\eqref{eq:psi-def}; $\psi$ is computed before any committee is chosen. 
\end{proof}

\begin{remark}[\sc A $\mathbf{z}$-aware robustness check]
\label{rem:t-aware}
One can define a party-level surrogate that discounts within-party lower-ranked candidates by an OWA vector $\tilde{\mathbf{z}}$:
for each vote, sort the $\mathbf{s}$-scores of party-$i$ candidates in nonincreasing order and aggregate them by $\tilde{\mathbf{z}}$; then sum over voters and normalize across parties as in~\eqref{eq:psi-def}. 
Our experiments (omitted here) show qualitatively similar conclusions; we keep the $\mathbf{z}$-independent $\psi$ as default because $\mathbf{s}$ captures \emph{voting},
while $\mathbf{z}$ captures \emph{within-committee aggregation}.
\end{remark}

\begin{proposition}[\sc Party-list ballots align $\psi$ with vote shares under SNTV]
\label{prop:psc-bridge}
With the plurality ($1$-approval) scoring vector $\mathbf{s}=\mathbf{u}_{1,m}$, the aggregate party score equals the party's voter share: $\psi_i^\mathrm{SNTV}=\frac{n_i}{n}$ for each $i\in P$, where $n_i$ is the number of voters whose top-ranked candidate is affiliated with party $i$.
\end{proposition}

\begin{proof}
Under $\mathbf{s}=\mathbf{u}_{1,m}$, a vote contributes $1$ only to the unique top-ranked candidate. Thus, summing over all voters, party $i$ receives exactly $n_i$ points in the numerator of~\eqref{eq:psi-def}. The denominator is $|\mathbf{s}|\|\mathbf{s}\|_1 = n\cdot 1 = n$, hence $\psi_i^\mathrm{SNTV} = n_i/n$.
\end{proof}

Allocation proportionality as a concept is much more akin to apportionment proportionality than to PSC or the various justified representation axioms. Most importantly, we do not really expect any voting rule to satisfy allocation proportionality, since we have the classic problem of seats being discrete and (usually) much less numerous than voters. Instead, we are thinking primarily in terms of a \emph{degree} of allocation proportionality.

A number of \emph{proportionality indices} have been developed in the field of apportionment proportionality \cite{TaageperaGrofman03}. Norm-based measures dominate the literature \cite[see, e.g.,][]{LoosemoreHanby71,Rose84,Gallagher91,Monroe94}, including  $L_1$ (Manhattan), $L_2$ (Euclidean), and $L_\infty$ (Chebyshev) distances between the vote share vector and the seat share vector. A big advantage of norm-based measures of proportionality is their easy interpretability. Moreover, they possess most of the axiomatic properties desirable in a proportionality index \cite{TaageperaGrofman03}.

\citet{Wada16} proposed $\alpha$-divergence as a theoretically superior alternative to norm-based proportionality indices \cite[see also][]{WadaKamahara18}. For $\mathbf{v}, \mathbf{s} \in \Delta_{d}$ and $\alpha \in \mathbb{R}$, $\alpha$-divergence of $\mathbf{s}$ from $\mathbf{v}$ is generally given by:
\begin{equation} \label{eq:alphaDiv}
    D^{\alpha}(\mathbf{v} \Vert \mathbf{s}) = \frac{1}{\alpha(\alpha-1)} \sum_{i=1}^{d} s_i \left(\left(\frac{v_i}{s_i}\right)^{\alpha} - 1\right),
\end{equation}
and for $\alpha = 0, 1$ by taking the limit:
\begin{equation} \label{eq:alphaDiv0}
    D^{0}(\mathbf{v} \Vert \mathbf{s}) = -\sum_{i=1}^{d} s_i \log\left(\frac{v_i}{s_i}\right),
\end{equation}
and
\begin{equation} \label{eq:alphaDiv1}
    D^{1}(\mathbf{v} \Vert \mathbf{s}) = \sum_{i=1}^{d} v_i \log\left(\frac{v_i}{s_i}\right).
\end{equation}
Note that \ref{eq:alphaDiv1} is the Kullback--Leibler \cite{KullbackLeibler51} divergence of $\mathbf{v}$ from $\mathbf{s}$. Note also that $\alpha$-divergence is not a distance measure, since it satisfies neither symmetry nor the triangle inequality. However, as in the case of distances, divergence of $0$ implies $\mathbf{v} = \mathbf{s}$.

The choice of the optimal measure of allocation proportionality degree depends on the particular application. When reporting the experimental results, we use the $\alpha$-divergence, $\alpha = 1$, between the aggregate party score vector and the seat share vector as our primary measure.

\begin{example}[\sc A small party election illustrating allocation proportionality]
\label{ex:running}
Consider two parties $A,B$ and $m=4$ candidates $\{A_1,A_2,B_1,B_2\}$ with committee size $k=3$ and $n=10$ voters. 
The preference profile is:
\begin{align*}
&3\times \bigl(A_1 \succ A_2 \succ B_1 \succ B_2\bigr),\quad
  3\times \bigl(A_2 \succ A_1 \succ B_1 \succ B_2\bigr),\\
&2\times \bigl(B_1 \succ B_2 \succ A_1 \succ A_2\bigr),\quad
  2\times \bigl(B_2 \succ B_1 \succ A_1 \succ A_2\bigr).
\end{align*}

\paragraph{Aggregate party scores $\psi$ (Def.~\ref{def:aps}).}
Take the Borda scoring vector on $m=4$, $b_4=(3,2,1,0)$.%
\footnote{Using the normalized Borda vector $b_4/(m-1)$ would not change any $\psi$ values,
because the normalization cancels in~\eqref{eq:psi-def} below.}
The candidate Borda totals are
$A_1=19$, $A_2=15$, $B_1=16$, $B_2=10$, hence the party totals are
$A:34$ and $B:26$. Normalizing by $\sum_{c\in C}\text{Borda}(c)=60$ gives
\[
\psi_A=\tfrac{34}{60}=0.566\overline{6},\qquad \psi_B=\tfrac{26}{60}=0.433\overline{3}.
\]
Formally, with notation from Definition~\ref{def:aps}, 
\begin{equation}
\psi_i \;=\; 
\frac{\sum_{\prec\in w}\sum_{j\in a^{-1}(i)} s_{\mathrm{pos}_\prec(j)}}{|w|\,\|s\|_1},
\end{equation}
which for $s=b_4$ specializes to the computation above.
\medskip

\paragraph{Seat shares under common OWA rules.}
For each rule we evaluate the committee score of all $\binom{4}{3}$ committees.
For $k$--Borda, Bloc, Harmonic Borda, and PAV the unique maximizer is 
$S^\star=\{A_1,A_2,B_1\}$.
For SNTV and CC there is a tie between $\{A_1,A_2,B_1\}$ and $\{A_1,A_2,B_2\}$,
which we resolve lexicographically in favor of $S^\star$.%
\footnote{Any fixed tie-breaking produces the same party seat shares in this example.}
Thus the observed seat shares are $s_A=2/3$ and $s_B=1/3$.

\paragraph{Deviation from allocation proportionality.}
Two standard indices are:
\[
L_1=\sum_i |\psi_i-s_i|=0.2,\]
and
\[
\Dkl{\psi}{s}=\sum_i \psi_i \ln\!\frac{\psi_i}{s_i}\approx 0.0216.
\]
Both indicate mild over-representation of $A$.
\end{example}

Another potentially useful measure of deviation from allocation proportionality is \emph{quota compliance}:

\begin{definition}[\sc $\psi$-quota compliance]
\label{def:quota}
Given a party election and a committee of size $k$, let $q_i$ be party $i$'s seat share and let $\psi_i$ be its aggregate score (Def.~14).
We say the outcome satisfies \emph{lower/upper $\psi$-quota} if for every $i\in P$,
\[
q_i \in \left\{ \frac{\lfloor k\,\psi_i \rfloor}{k},\ \frac{\lceil k\,\psi_i \rceil}{k} \right\}.
\]
Over a distribution of instances, we can calculate the fraction of elections in which \emph{all} parties satisfy this condition (strong compliance) as well as the per-party compliance rate (weak compliance).
\end{definition}

To better understand the political effects of deviations from allocation proportionality, we use two further measures: \emph{seat bias} and \emph{defractionalization} (change in the \emph{effective number of parties}). By seat bias we understand the difference between the aggregate party score and that party's seat share. It simply measures whether the party got more seats than it should -- or instead less of them. We are primarily interested in measuring how bias correlates with party size. However, since parties are exchangeable, average bias over, say, the first party should likewise be $0$, since it should be large just as often as it is small. We solve the problem by sorting the parties before averaging, thus obtaining average bias measures for the $i$-th largest party, $i = 1, \dots, |P|$.

\emph{Effective number of components} is a measure of concentration frequently used in political science, where it has been introduced by \citet{LaaksoTaagepera79}. Dating back to a diversity index proposed independently by \citet{Simpson49} in ecology, as well as by \citet{Hirschman45} and \citet{Herfindahl50} in economics, it is the inverse of the sum of squared coordinates of a vector of component sizes:

\begin{definition}[\sc Effective Number of Components]
    For a weight vector $\mathbf{x} \in \mathbf{\Delta}_{d}$, the effective number of components is given by:
    \begin{equation}
        \eta(\mathbf{x}) = \left(\sum_{i=1}^{d} x_i^2 \right)^{-1}
    \end{equation}
\end{definition}

There are two ways of using that index for measuring the effective number of parties: we can calculate the \textit{a priori} number of effective parties, using the vector of vote shares (or, in our case, aggregate scores), and the \textit{a posteriori} number of effective parties, using the vector of seat shares. By itself, neither of those measures proportionality. However, a ratio of the \textit{a posteriori} to \textit{a priori} effective number of parties can be used as an auxiliary measure of allocation proportionality, since it tells us whether a voting rule amplifies or tapers party concentration arising from voter preferences.

\begin{remark}[\sc Evaluating non-OWA rules]
\label{rem:non-owa}
Definition~15 compares seat shares to $\psi^{(s)}$, which depends only on a chosen scoring vector $s$.
Hence, after fixing an \emph{evaluation} scoring vector $\mathbf{s}^\star$ (e.g., Borda or plurality), we can compute $\psi^{(\mathbf{s}^\star)}$ once per profile and evaluate any multiwinner rule (including STV) by comparing its seat shares to $\psi^{(\mathbf{s}^\star)}$ via the indices from Sec.~5.
This separates the choice of a rule from the choice of a vote surrogate.
\end{remark}

\section{Experiments}

Even for multi-district apportionment, proportionality is difficult to model analytically \cite[see, e.g.,][]{BoratynEtAl25b}. Thus, like others \citep[e.g.,][]{BardalEtAl25}, we study the problem primarily experimentally. 
For each combination of the following three parameters: statistical culture (see below), number of parties ($|P| = 3, 4, 5, 6, 8, 10$), and committee size (per district) ($k = 1, 2, 3, 4, 8, 12, 16, 24$), we run $256$ \emph{multi-district experiments}. For experiments with single-winner districts, we have district count $128$. For $2$-winner districts, we have district count $64$; for $3$-winner districts, district count $48$; and for all other district sizes, district count $32$, ensuring in each case that we always have at least $128$ seats and therefore the discretization effects are minimized. In each district, there are $n = 1024$ voters and $k$ candidates per party.

\subsection{Statistical Cultures}

Throughout the paper, we have assumed that allocation proportionality can avoid the problems with identification of cohesive voting groups because such groups are predefined as parties. Yet for that assumption to be true with respect to experimental elections, we need to ensure that the probabilistic model for generating voter preferences does not treat each party's candidates as independent from one another, but instead groups them, reflecting the background assumption that candidates of the same party should be perceived by voters as, on average, more similar than candidates of different parties. Clearly, this is not a feature of any of the probabilistic models used in computational social choice \cite[see generally][]{BoehmerEtAl24,SzufaEtAl20}, which have been developed for non-party elections. Thus, we need to extend them.

Since we are using multi-district party elections for our experiments, candidate group is not the only adjustment we need to make. In real-life district elections, we encounter \emph{intra-district voter clustering}. Voters within a single district share preferences to a greater extent than an unbiased sample of the population \cite{CliffOrd73}. Again, our probabilistic models should reflect that. Otherwise the $c$-district seat shares would (per the central limit theorem) converge to the expected value for the population profile.

Both of those problems have already been addressed by \citet{BoratynEtAl22}. Following them, we consider four classes of probabilistic models:

\begin{definition}[\sc Spatial Models]
    In a $d$-dimensional Euclidean model each party, voter, and candidate is assigned an ideal point in $\mathbb{R}^{d}$ \citep{EnelowHinich84,Merrill84}. First, party ideal points are drawn from $\operatorname{Unif}((0, 1)^{d})$, and then candidate ideal points for each district are drawn from the multivariate normal distribution with location at the party's ideal point and the covariance matrix $\mathbf{\Sigma} := \sigma I_{d}$, where $I_{d}$ is a $d \times d$ identity matrix and $\sigma \in \mathbb{R}_{+}$. Voter ideal points are drawn independently from the uniform distribution on $(0, 1)^{d}$, then shifted in each district independently by a vector drawn from $\operatorname{Unif}(-1/4, 1/4)^{d}$ (to account for district clustering). A vote is obtained by sorting candidates according to the increasing $L_2$ distance from the voter's ideal point.
\end{definition}

\begin{definition}[\sc Single-Peaked Models]
    We consider two models for generating single-peaked profiles. In one \citep{Walsh15a} we are given an ordering on the set of candidates, and each vote is drawn from the uniform distribution on the set of all single-peaked votes consistent therewith. In the other one \citep{Conitzer09} the peak is drawn from a uniform distribution on candidates, and the remainder of the vote is obtained by a random walk. In both models, candidate ordering in each district is obtained as in the $1$-dimensional Euclidean model.
\end{definition}

\begin{definition}[\sc Mallows Model]
    The Mallows model \citep{Mallows57,CritchlowEtAl91} is parametrized by a single parameter $\phi \in [0,1]$, and a (central) vote $\mathbf{v}_c \in L_C$. The probability of generating a vote $\mathbf{v}$ is proportional to $\phi^{f(\mathbf{v}_c, \mathbf{v})}$, where $f(\mathbf{v}_c, \mathbf{v})$ is the Kendall tau distance \citep{Kendall38} between $\mathbf{v}_c$ and $\mathbf{v}$, i.e., the minimum number of swaps of adjacent candidates needed to transform the vote $\mathbf{v}$ into the central vote $\mathbf{v}_c$.
    We first generate a central vote for each district, $\mathbf{v}_c^i$ with parameter $\phi_{1}$ and a starting vote $\mathbf{v}_c^0 := [m]$, then generate votes within each district with parameter $\phi_{2}$ and a district--wide central vote $\mathbf{v}_c^i$. Intra-party clustering is achieved by grouping party candidates together in the starting vote. On sampling from the Mallows model, see \cite{LuBoutilier14}~\footnote{We use a Mallows model parameterization by Boehmer et al. \cite{BoehmerEtAl21}, based on a normalized dispersion parameter $\normphi$.}.
\end{definition}

\begin{definition}[\sc Impartial Culture (IC)]
    Under IC, each vote is drawn randomly from the uniform distribution on linear orders on candidates \citep{CampbellTullock65}. This model does not account for intra-party and intra-district clustering, and is regarded as a poor approximation of real-life \citep{RegenwetterEtAl06,TidemanPlassmann12}.
\end{definition}

\subsection{Results}

We focus on two models that resemble political elections most closely: 2-D Euclidean and 1-D Euclidean. The former is likely more realistic, while the latter is of particular interest because 1-D Euclidean elections are single-peaked. The results for other models (Mallows, Conitzer, Walsh, and IC) are in the forthcoming Appendix. For each model, we look at five statistics: $\alpha$-divergence, $L_2$ norm, defractionalization, and bias for the largest and for the smallest party.

\begin{figure}[!htb]
    \centering
    \begin{subfigure}{0.95\textwidth}
        \includegraphics[width=\textwidth]{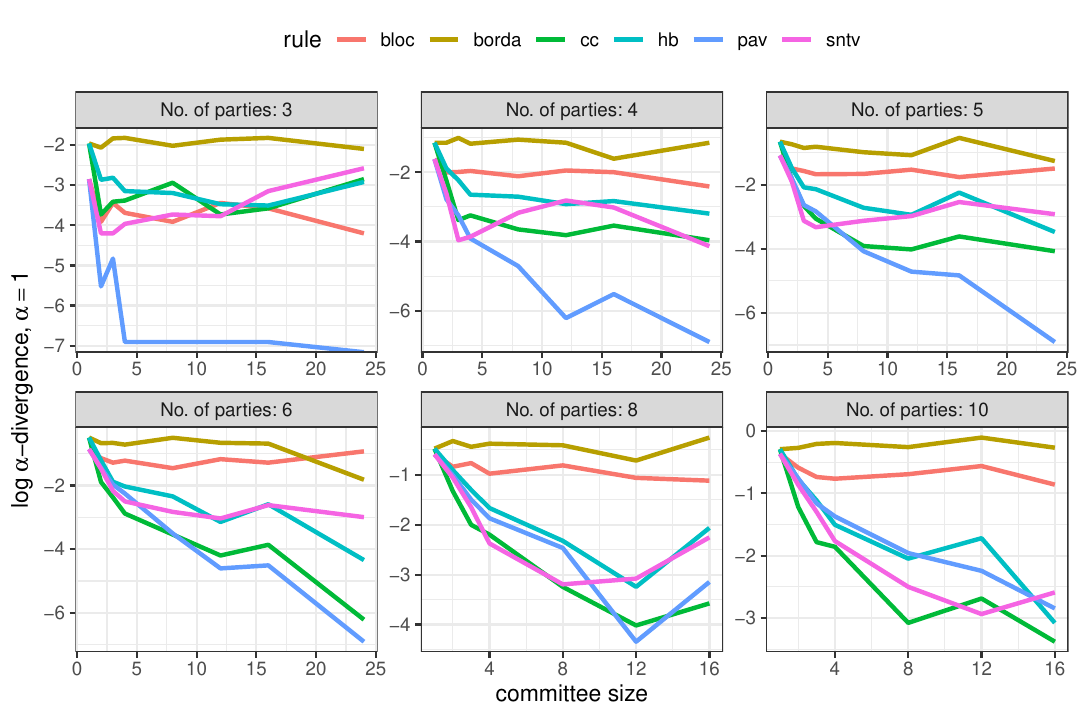}
        \caption{1-D Euclidean}
    \end{subfigure}
    \begin{subfigure}{0.95\textwidth}
        \includegraphics[width=\textwidth]{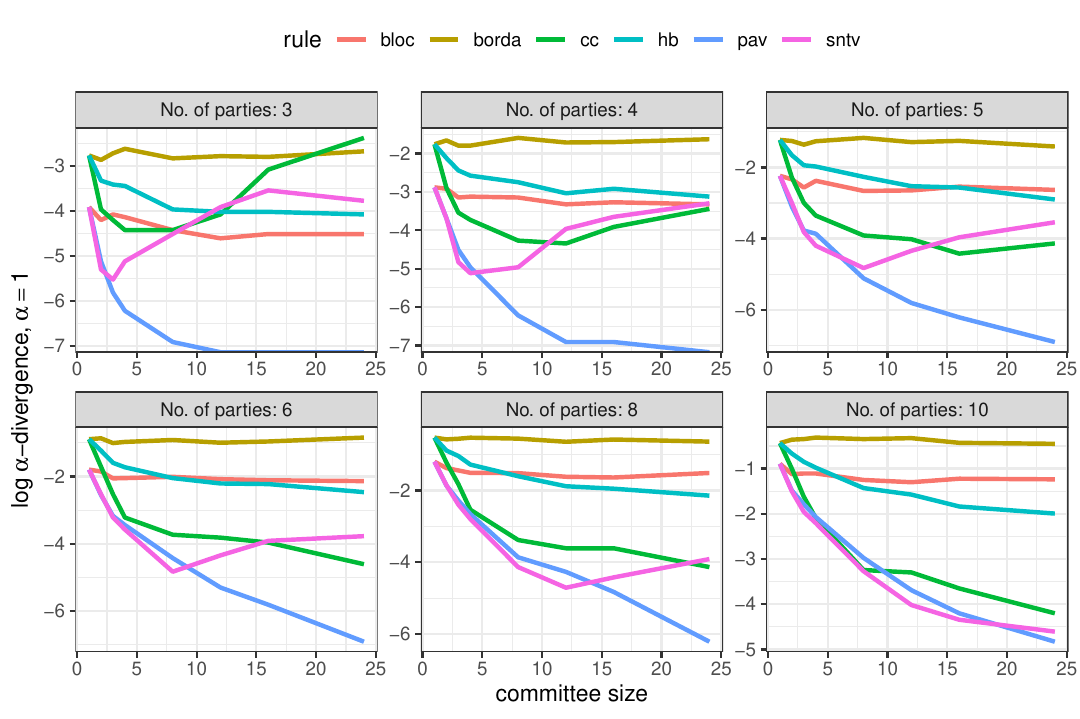}
        \caption{2-D Euclidean}
    \end{subfigure}
    \caption{Allocation proportionality degree ($\alpha$-divergence).}
    \label{fig:alphaDiv}
\end{figure}

\begin{figure}[!htb]
    \centering
    \begin{subfigure}{0.95\textwidth}
        \includegraphics[width=\textwidth]{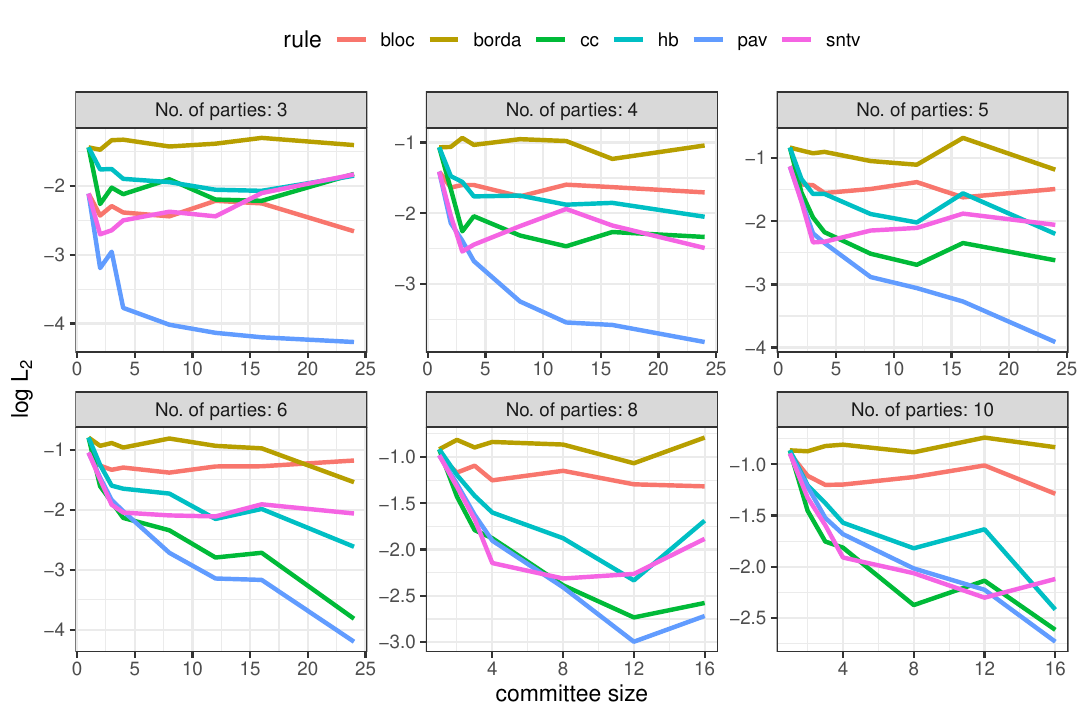}
        \caption{1-D Euclidean}
    \end{subfigure}
    \begin{subfigure}{0.95\textwidth}
        \includegraphics[width=\textwidth]{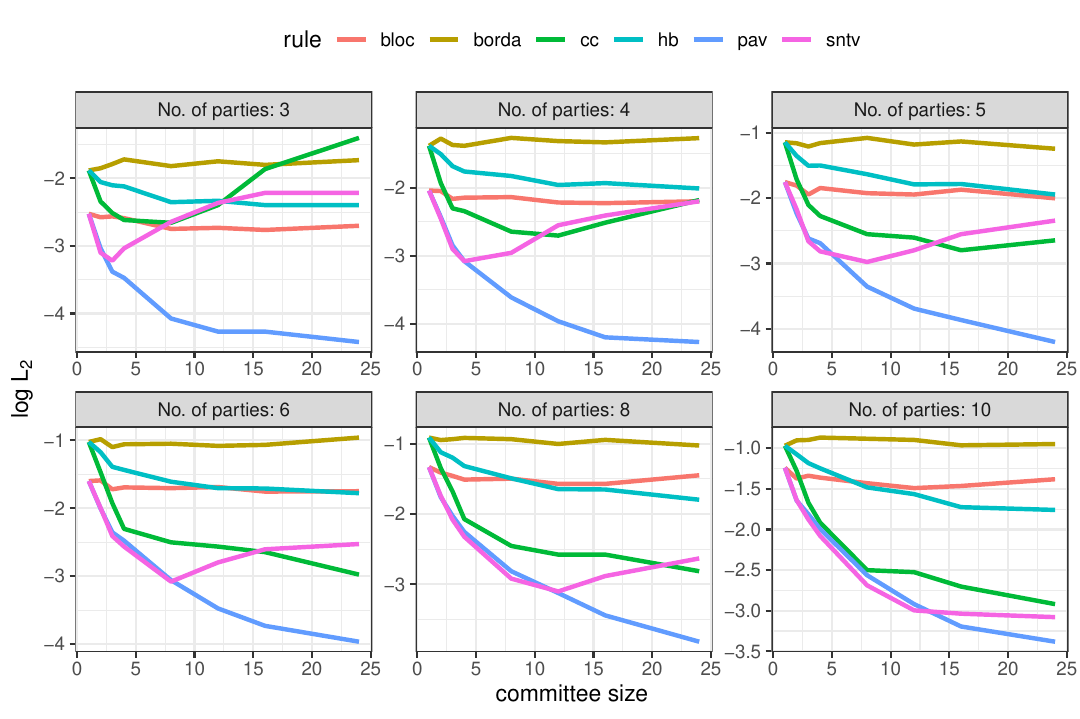}
        \caption{2-D Euclidean}
    \end{subfigure}
    \caption{Allocation proportionality degree ($L_2$ metric).}
    \label{fig:l2}
\end{figure}

\begin{figure}[!htb]
    \centering
    \begin{subfigure}{0.95\textwidth}
        \includegraphics[width=\textwidth]{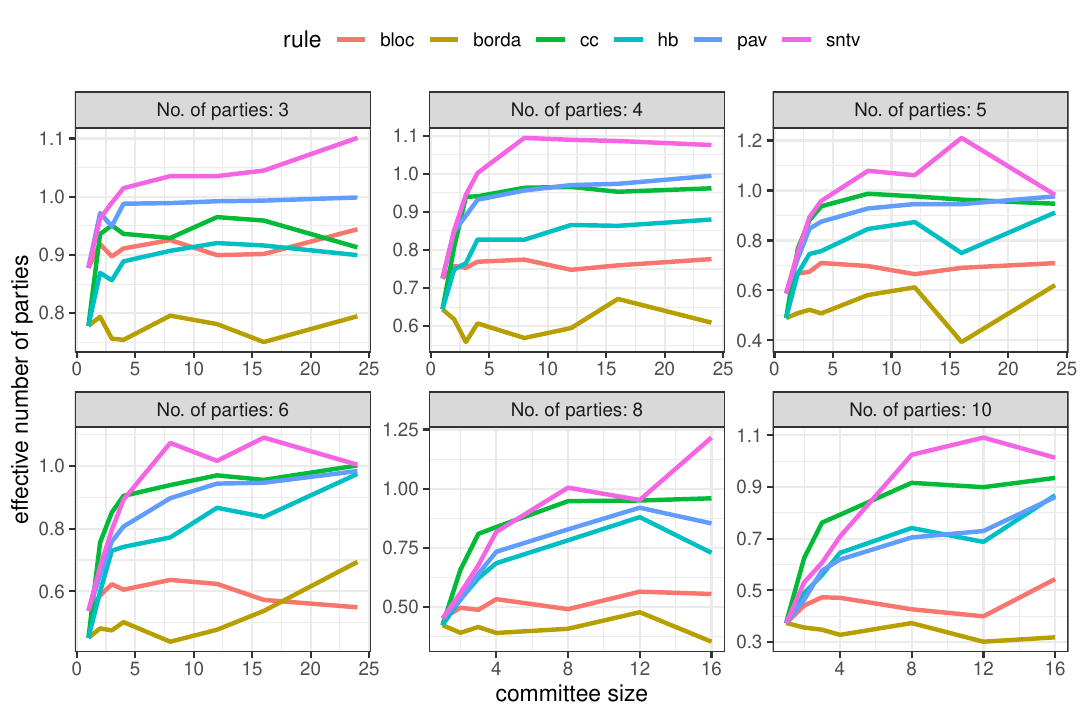}
        \caption{1-D Euclidean}
    \end{subfigure}
    \begin{subfigure}{0.95\textwidth}
        \includegraphics[width=\textwidth]{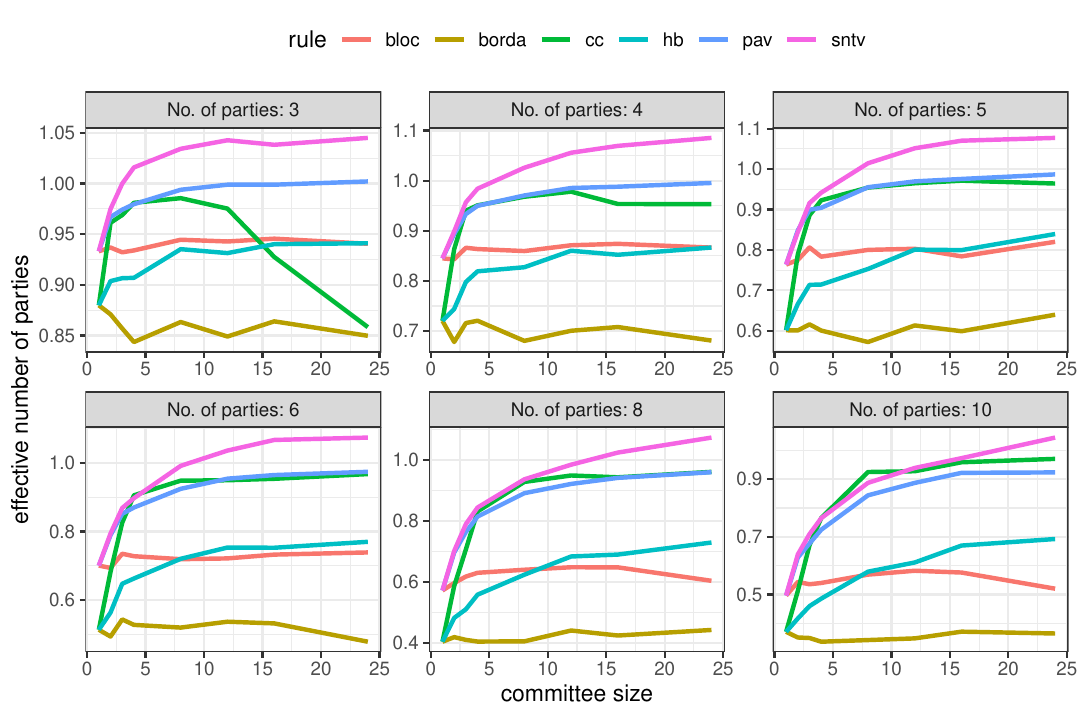}
        \caption{2-D Euclidean}
    \end{subfigure}
    \caption{Defractionalization (change in the effective number of parties).}
    \label{fig:effChange}
\end{figure}

\subsection{Discussion}

For three rules, we have a clear ordering. First, $k$--Borda is usually the least allocation proportional rule, being biased in favor of the larger parties and strongly amplifying concentration. Bloc Voting exhibits similar properties, albeit on a reduced scale. On the other hand, PAV is usually the most allocation proportional rule and does not exhibit significant bias. PAV's allocation proportionality degrees improve with committee size, while for $k$--Borda this parameter appears to have no effect. For BV, a modest improvement with committee size usually exists, but the trend is weaker and sometimes close to flat or slightly non-monotone, depending on the number of parties and the model.

For other rules, the results are more complex. For Harmonic Borda, the rule's relative performance depends on the number of parties and the model used. For small number of parties, especially under 2-D models, it can be even more disproportional than Bloc Voting, though not as disproportional as \linebreak $k$--Borda. For Chamberlin--Courant, the picture is even more complex: for large number of parties it can be highly proportional (not as proportional as PAV, but nearly so). However, its proportionality depends non-monotonically on district magnitude: there exists a minimum (whose position depends on the number of parties). For higher district magnitudes, CC's proportionality can rapidly decrease (especially under \linebreak 2-D models) in the direction of amplified concentration of parties. Finally, for SNTV we are also seeing a number of parties-dependent minimum, but above that minimum, SNTV becomes more egalitarian than proportional, being biased against larger parties and in favor of smaller ones.

\section{Future Work}

Future work will focus on several issues. First, we plan a more systematic study of OWA--based voting rules aimed at determining how allocation proportionality is affected by the OWA vector. Second, we will investigate the relation between allocation proportionality and axiomatic notions of proportionality such as PSC, seeking analytical results. Third, we will compare how allocation proportionality degree relates to other quantitative measures of proportionality, such as the PSC value \cite{BardalEtAl25}. Fourth, we will investigate the non-monotonic behavior of certain common rules, such as SNTV and Chamberlin--Courant. Finally, we plan to extend the proposed framework to other committee scoring rules, and to ordinal voting rules more generally.

\FloatBarrier

\begin{figure}[!htb]
    \centering
    \begin{subfigure}[b]{0.5\textwidth}
        \includegraphics[width=\textwidth]{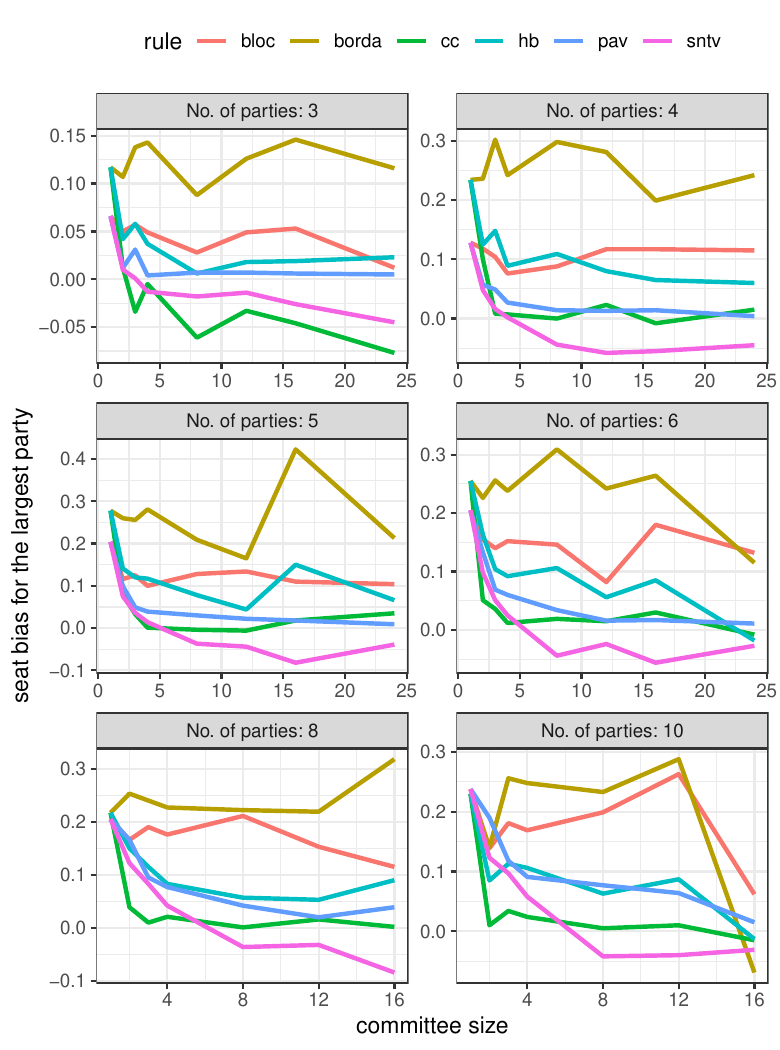}
        \caption{1-D Euclidean}
    \end{subfigure}
    ~
    \begin{subfigure}[b]{0.5\textwidth}
        \includegraphics[width=\textwidth]{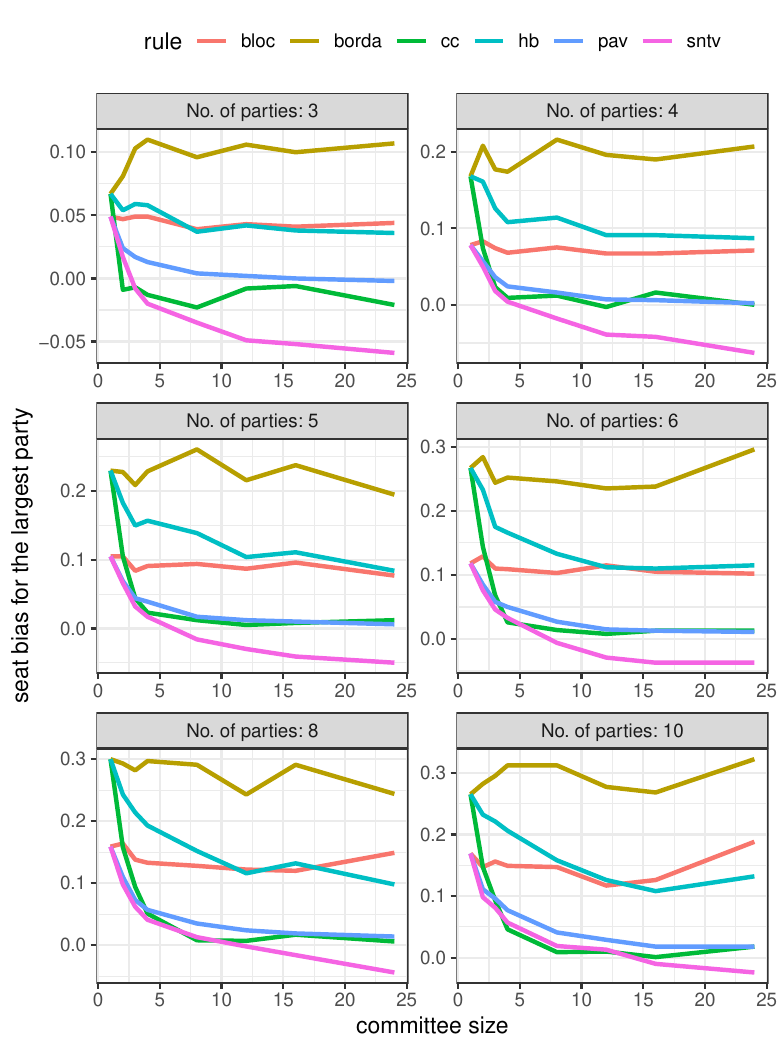}
        \caption{2-D Euclidean}
    \end{subfigure}
    \caption{Bias for the largest party.}
    \label{fig:bias}
\end{figure}

\FloatBarrier

\bibliographystyle{ACM-Reference-Format}
\bibliography{spoilers}


\end{document}